\documentclass[twocolumn,3p,times,procedia]{elsarticle}

\usepackage{ecrc}
\usepackage{graphicx}
\usepackage{titlesec}
\usepackage{xcolor}
\volume{00}
\firstpage{1}
\runauth{}
\jid{procs}
\CopyrightLine{2015}{Published by Elsevier Ltd.}
\usepackage[english]{babel}
\usepackage{amssymb,amsmath,amsthm}
\newtheorem{theorem}{Theorem}
\usepackage{multicol}
\usepackage{float}
\usepackage[figuresright]{rotating}
\usepackage{bm}
\usepackage{caption}
\usepackage{pifont}
\newcommand{\xmark}{\ding{55}}%
\usepackage{fancyhdr}
\usepackage{threeparttable} 
\fancyheadoffset[RO,EL]{0pt}
\usepackage{geometry}

\begin{document}

\begin{frontmatter}

\dochead{}
\title{
\begin{flushleft}
{\bf \Huge EBAKE-SE: A Novel ECC Based Authenticated Key Exchange between Industrial IoT Devices using Secure Element}
\end{flushleft}
}
\author[]{\bf \Large \leftline {Chintan Patel$^a$, Ali Kashif Bashir$^b$, Ahmad Ali AlZubi$^c$, Rutvij H Jhaveri$^*$$^d$}}

\address{\bf  \leftline {$^a$School of Artificial Intelligence, Information Technology and Cyber Security, Rashtriya Raksha University, Gujarat, India}

\bf  \leftline {$^b$Department of Computing and Mathematics, Manchester Metropolitan University, United Kingdom}
\bf  \leftline {$^c$Computer Science Department, Community College, King Saud University,  Saudi Arabia}
\bf  \leftline {$^d$Department of Computer Science and Engineering, School of Technology, Pandit Deendayal Energy University, Gujarat, India}
}
\fntext[]{Dr. Chintan Patel is an Assistant Professor in Rashtriya Raksha University which is an institute of national importance in india. (email:chintan.p592@gmail.com:).}

\fntext[]{Dr. Ali Kashif Bashir is associated with Manchester Metropolitan University, United Kingdom (email:dr.alikashif.b@ieee.org).}

\fntext[]{Dr. Ahmad Ali AlZubi is associated with King Saud University, Saudi Arabia (email:aalzubi@ksu.edu.sa).}

\fntext[]{Dr. RutviJ H. Jhaveri  (Corresponding author)  is an Assistant Professor in Pandit Deendayal Energy university. (email:rutvij.jhaveri@gmail.com:).}

\begin{abstract}
Industrial IoT (IIoT) aims to enhance services provided by various industries such as manufacturing and product processing. IIoT suffers from various challenges and security is one of the key challenge among those challenges. Authentication and access control are two notable challenges for any Industrial IoT (IIoT) based industrial deployment. Any IoT based Industry 4.0 enterprise \textcolor{black}{designs} networks between hundreds of tiny devices such as sensors, actuators, fog devices and gateways. Thus, \textcolor{black}{articulating a secure authentication protocol between sensing} devices or a sensing device and user devices is an essential step in IoT security. In this paper, first, we present cryptanalysis for the certificate-based scheme proposed for similar environment by Das et al. and prove that their scheme is vulnerable against various traditional attacks such as device anonymity, MITM, and DoS. We then put forward an inter-device authentication scheme using an \textcolor{black}{ECC (Elliptic Curve Cryptography)} that is highly secure and lightweight compared to other existing schemes for a similar environment. Furthermore, we set forth a formal security analysis using random oracle based ROR model and informal security analysis over the Doleve-Yao channel. In this paper, \textcolor{black}{we presents comparison of proposed} scheme with existing schemes based on communication cost, computation cost and security index to prove that the proposed EBAKE-SE is highly efficient, reliable, and trustworthy compared to other existing schemes for an inter-device authentication. At long last, we present an implementation for the proposed EBAKE-SE using MQTT protocol. 
\end{abstract}


\end{frontmatter}
\section{Introduction}
\noindent Industrial Internet of Things (IIoT) network is built up using a highly homogeneous, globally dynamic, deeply deployed, and comparatively resource-constrained devices to provide \textit{"Any type"} service at \textit{"Any location"} to \textit{"Anyone"} on \textit{"Any time"} \cite{wang2021sparse} \cite{extra4}. \textcolor{black}{The} Scale of IIoT data generation is directly proportional to the growing quantity of internet-connected devices. As per recent predictions (June 2019) by the global giant of telecommunications and market intelligence agency International Data Cooperation (IDC), there will be approx 42 billion deployed devices that will generate approx 80 ZettaByte data by 2025 \cite{jhaveri2021fault} \cite{extra3}. 

IIoT-based devices are a mixture of resource-constrained devices as well as resource-capable devices. \textcolor{black}{devices} deployed on the ground (ex. home for the smart home, factory for the smart industry, road for the smart transportation) are mostly resource-constrained devices such as sensors and actuators. Devices which collect data from these sensing devices (\textcolor{black}{A.c.a} gateway devices) are hybrid devices such as routers, raspberry-pi and node-MCU. The edge device or the fog device receives unstructured data from the sensing devices and performs pre-processing on that data to convert it into structured data. These edge devices are resource-capable and forward only necessary structured data over the cloud or to the user \cite{ali2019fractal}. Edge devices reduce unnecessary traffic over the cloud server through their intelligent pre-processing.

\begin{figure}
\includegraphics[width=\linewidth]{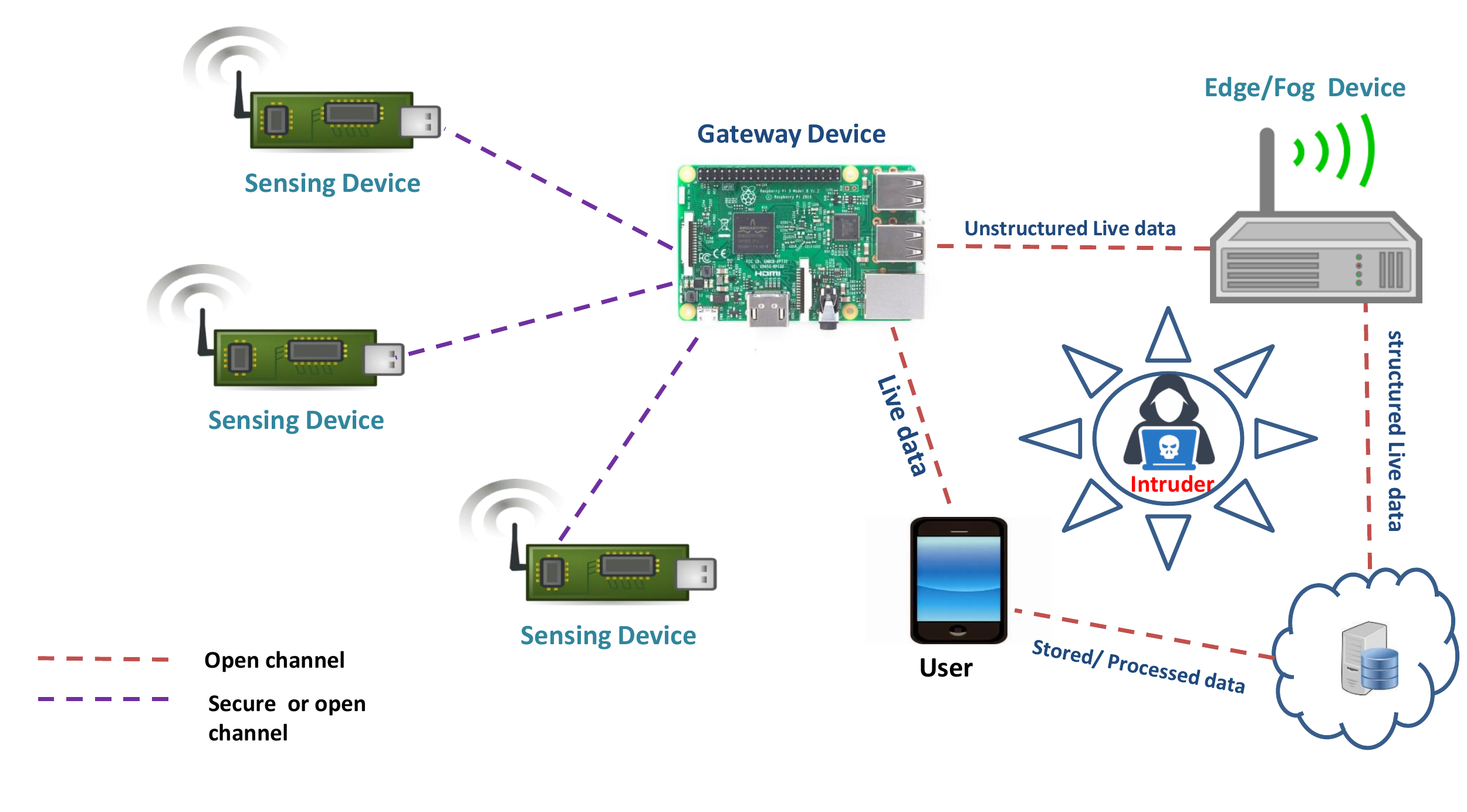}
\caption{Inter device data transfer in IoT.} \label{fig1}
\end{figure}

Fig. \ref{fig1} presents an overview \textcolor{black}{for} the generic IIoT \textit{"data chain"}. It highlights how raw material (i.e., unstructured data) is converted into the smart product (i.e., knowledge) used for quick and accurate decision-making. The IoT ecosystem consists of three significant aspects. (1) IoT devices (2) reliable, optimized, and secure communication between devices (3) data processing and knowledge generation. A recent survey by Sobin \cite{Sobin2020} highlights that scalability, lack of standard architectures and protocols, energy efficiency, and security and privacy are still open issues that limit the wide-range deployment of an IoT ecosystem. Other past surveys \cite{Atzori2010}, \cite{Fuqaha2015}, \cite{Sethi2017},\cite{extra1}, \cite{extra2} also highlighted that the IoT ecosystem suffers from the numerous privacy and security issues due to its resource-constrained devices, heterogeneous deployment, and dynamic nature. 

In the recent past, authors in \cite{Atzori2010}, \cite{Maple2017}, \cite{stoyanova2020survey}, \cite{irshad2020fuzzy}, \cite{MOHAMADNOOR2019} presented a brief study on numerous challenges and issues related to the IoT security and privacy. Author highlights an \textit{"authentication"} as a common threat to the IoT ecosystem. A secure and reliable authentication defines as a mutual trust-building between user-device and device-device through a resource-efficient key exchange protocol \cite{arul2018console}. In this paper, we provide cryptanalysis for the scheme proposed by Das et al. \cite{Das2019} for device-to-device authentication in a similar environment. We highlight \textcolor{black}{that the scheme proposed by Das et al.} is vulnerable to numerous attacks such as device impersonation, Man-In-The-Middle (MITM) attack, and Denial of Service (DoS) attack. We then put forward a considerably reliable and efficient inter-device Remote user authentication (RUA) scheme using a secure element (SE) and an Elliptic Curve Cryptography (ECC). Recently Qureshi et al. presented stream-based authentication for big data networks based on IoT sensing devices \cite{qureshi2021stream}        

\vspace{0.1in}
\noindent \textit{\textbf{Contribution:}} In this paper,
\begin{itemize}
    \item We \textcolor{black}{present} cryptanalysis for the authentication scheme proposed by Das et al. \cite{Das2019} for device-to-device authentication. We prove their scheme is not secured against device impersonation, MITM, and a DoS attack.
    \item We \textcolor{black}{present} a novel authentication scheme between two smart IIoT devices via Trusted Authority (TA) using ECC and SE.
    \item We present an informal security analysis for the proposed scheme using \textit{send} and \textit{receive} based Dolev-Yao channel. We then offer a formal security analysis for the proposed EBAKE-SE using a random oracle-based challenge-response game.
    \item Next, we \textcolor{black}{demonstrate} the implementation scenario and real-time results for the proposed EBAKE-SE using the physical IIoT devices.
    \item Furthermore, We put forward a comparative analysis of the proposed work with an existing work based on time and space requirements.
\end{itemize}

\noindent \textit{\textbf{Case study and Motivations:}}
IoT is a complex matrix of the numerous resource-constrained devices, as well as countless resourceful Advanced IoT (AdIoT) devices \textcolor{black}{\cite{Sethia2018}}. The internet-connected \textit{smart home} appliances such as a washing machine, refrigerator, AC, and CCTV system are considered an AdIoT devices. Wearable devices such as a smartwatch and a smart belt (for health monitoring) are lightweight, resource-constrained devices. Recent surveys show that 98\% of IoT devices communicate over an open channel, which is the most significant threat to a person's privacy and the confidentiality of data. The \textit{smart healthcare} system is equipped with numerous remotely controlled devices such as an intelligent ventilator, smart oxygen supplier, and a smart patient monitoring system. The prosperous attack on these devices can create chaos in the healthcare system. Thus, defending these IoT devices from traditional vulnerabilities and attacks is highly desirable. Any IoT system must ensure data confidentiality, data integrity, user privacy, secure device authentication, and secure device access control. Protecting the IoT devices from attacks such as \textit{DoS, MITM, spoofing, and impersonation} is challenging for security professionals. It is profoundly anticipated that the IoT system users must not use traditional passwords and update them frequently. They must upgrade their system periodically and configure the latest security patches for their devices to protect them from the ransomware attacks such as \textit{WannaCry} and \textit{NotPetya}.     

\vspace{0.1in}
\noindent \textit{\textbf{Road map of the Paper:}} Section \ref{related} provides a brief overview of the recent related work to the proposed EBAKE-SE and basic preliminaries used for articulating this manuscript. In Section \ref{review}, we overview the scheme proposed by Das et al., followed by a cryptanalysis of Das et al.'s scheme in Section \ref{cryptanalysis}. In Section \ref{proposed}, we proposed a reliable and efficient device-device authentication scheme between two smart IoT devices using a TA. Section \ref{informal} and Section \ref{formal} presents formal and informal formal security analysis for the proposed EBAKE-SE respectively. Section \ref{implelment} discusses implementation for the proposed EBAKE-SE. In Section \ref{comparative}, we compare the proposed scheme with other schemes based on communication and computation costs. Finally, we conclude this paper in section \ref{conclusion}.

\section{Related work and Preliminaries}
\label{related}
\noindent In this section, we discuss related work to the proposed work and primary preliminaries required for articulating this paper.
\subsection{{Related Work}}
\label{sub:relatedwork}
Authentication creates trust among communicating devices \cite{patel2020enhanced}. An ECC came forward as an efficient and reliable advancement for lightweight cryptography. The ECC provides equally strong security compared to the RSA and other traditional methods in much lighter ways (smaller key size and addition-based discrete logarithm). An ECC plays a key role in the optimized deployments of lightweight cryptography. ECC is public-key cryptography that works on the base assumption that it is infeasible to find a discrete logarithm for the random elliptic curve element based on a publicly known base point. Miller introduced the use of ECC in 1985 \cite{miller1985use} and populated by koblitz in 1987 \cite{koblitz1987elliptic}. Between 1987 and 2021, numerous authors proposed the ECC-based key exchange and authentication schemes.

In 2019, Dhillon et al. \cite{dhillon2019secure} proposed an ECC based authentication scheme for the \textit{SIP (Session Initiation Protocol)} that is used in \textit{VoIP (Voice-over-IP))} communication and provided a security analysis using AVISPA tool. Wearable devices play a key role in the numerous IoT-based applications such as smart healthcare and smart home. In 2019, Kumar et al. \cite{kumar2019secure} proposed the key exchange protocol between a user device (mobile device) and a wearable device using an ECC. In 2019, Lohachab et al. \cite{lohachab2019ecc} presented a scheme using an ECC for the MQTT communication and provided a security analysis using the AVISPA and an \textit{ACPT (Access Control Policy Testing)} tool. In 2019, Qi et al. \cite{qi2019secure} proposed an ECC-based authentication scheme for the secure session key establishment between a system user, \textit{Low Earth Orbit (LEO) satellite}, and the gateway device. 

In 2019, Garg et al. \cite{garg2019towards} also proposed an authentication scheme for the IIoT environment using lightweight operations, such as ECC and \textit{Physically Unclonable Functions (PUF)}. In 2019, Dammak et al. \cite{dammak2019token} proposed the token-based authentication scheme for the User-Gateway-device communication and claimed that their scheme is secured against a token impersonation attack and a stolen verifier attack. Recently Dang et al. \cite{dang2020pragmatic} proposed an authentication scheme using an ECC for the smart city environment. Authors in \cite{dang2020pragmatic} used Device-Device-Server \textit{(SD-SD-S)} network model for articulating their scheme and claimed that the proposed work achieves high energy efficiency. 

Designing a fully secured and highly resource-efficient security mechanism for \textcolor{black}{an} IoT environment is a challenging job. The IoT environment suffers from numerous vulnerabilities, such as inadequate physical security of the sensing devices, heterogeneity of the device manufacturers, proper standardizations, lower device synchronizations, and open ground for attackers. Hence, this paper proposes a novel authentication scheme that provides a robust and secured environment for session key generation between two IoT devices.      

\subsection{Preliminaries}
\label{sub:preliminaries}
\subsubsection{\textit{Elliptic Curve Cryptography (ECC)}}
\noindent An ECC provides a lightweight implementation for the public-key cryptography protocols such as a RSA with an equal level of the security. We can define an elliptic curve as a cubical curve of the form  \textit{$E_z$($\alpha,\beta$)} with the non repeatable roots defined over a finite field \textit{$\mathcal{F}_z$} where \textit{z} is a large prime number. We can represent an elliptic curve as per following Eq. \ref{eq1}.
\begin{equation}
    \label{eq1}
    E_z(\alpha,\beta) : Q^{2} = (P^{3} + \alpha*P + \beta) mod \gamma
\end{equation}
Here, \textit{P} and \textit{Q} are two curve points denoted by \textit{$P_t$(P,Q)}. The $\gamma$ represents a large prime number. Two constants \{$\alpha$, $\beta$\} are selected such that \{$\alpha$,$\beta$\} $\in$ \textit{$\mathcal{F}_z$} and their values must satisfy,
\begin{equation}
    4*\alpha^3 + 27*\beta^2 \neq 0 mod \gamma
\end{equation}
We can define the scalar point multiplication operations of an ECC over a point \textit{$P_t$} as $n*P_t$ = $P_t$ + $P_t$ +.......+ $P_t$ for \textit{n} times. The security of an ECC lies in finding the value of a large prime \textit{n} from the given $P_t$ and $n*P_t$. We can define the Elliptic Curve Discrete Logarithm Problem (ECDLP) as : from the given \textit{R = n*T}, it is difficult to find an integer \textit{n} in a polynomial time where \textit{n} $\in$  \textit{$\mathcal{F}_z$} and \textit{R} and \textit{T} are two points on elliptic curve \textit{$E_z$($\alpha,\beta$)}. We can define the Elliptic Curve Diffie-Hellman Problem (ECDHP) as: consider \{$\alpha$,$\beta$\} $\in$ \textit{$\mathcal{F}_z$} and \textit{P} is a point on the curve \textit{$E_z$($\alpha,\beta$)}. From the given $P$, $\alpha*P$ and $\beta*P$, it is difficult to compute a value $\alpha*\beta*P$ over \textit{$E_z$($\alpha,\beta$)} in a polynomial time.
\subsubsection{\textit{One-way Hash Function}}
\noindent A cryptographic hash function can be presented as $h:\{0,1\}^{*}\rightarrow{}\{0,1\}^{n}$ that takes string $p \in \{0,1\}^{*}$ as an input and outputs a fixed size binary string $Q\in\{0,1\}^{n}$. The cryptographic hash function must be collision resistant and preimage resistant for variable-size input and fixed-size output with enough randomness. 

\subsubsection{\textit{Network Model}}
\noindent A network model shown in Fig. \ref{fig2} \cite{Das2019} is followed for designing of authentication scheme. We consider the cloud trusted authority (TA) as a master controller in this network model. The IoT devices transmit data to each other over an open channel via the TA. The TA is a cloud MQTT server equipped with a broker. The IoT devices (such as a smart fridge or a gateway device) have \textcolor{black}{secure element that stores secret credentials in the tamper-resistant environment and the Wi-Fi module} (to connect with the internet). The secure element of a first device performs cryptographic operations in a tamper-proof environment and passes its outcome to the Wi-Fi module. This module publishes that data to the TA using the MQTT protocol, and the TA performs authentication operations and communicates with the second device using an MQTT. In this way, each of the three entities mutually authenticates each other, and after completion of the authentication phase, the IoT devices generate a one-time secure session key. Many authors follow another network model \cite{li2017robust} in that the gateway device is considered a trusted device due to the absence of a separate TA. Still, for the proposed scheme, we consider the presence of a separate TA (also works as a gateway) that setups security parameters for the IoT devices, including a gateway device, if required. 
\begin{figure}[H]
\includegraphics[width=\linewidth]{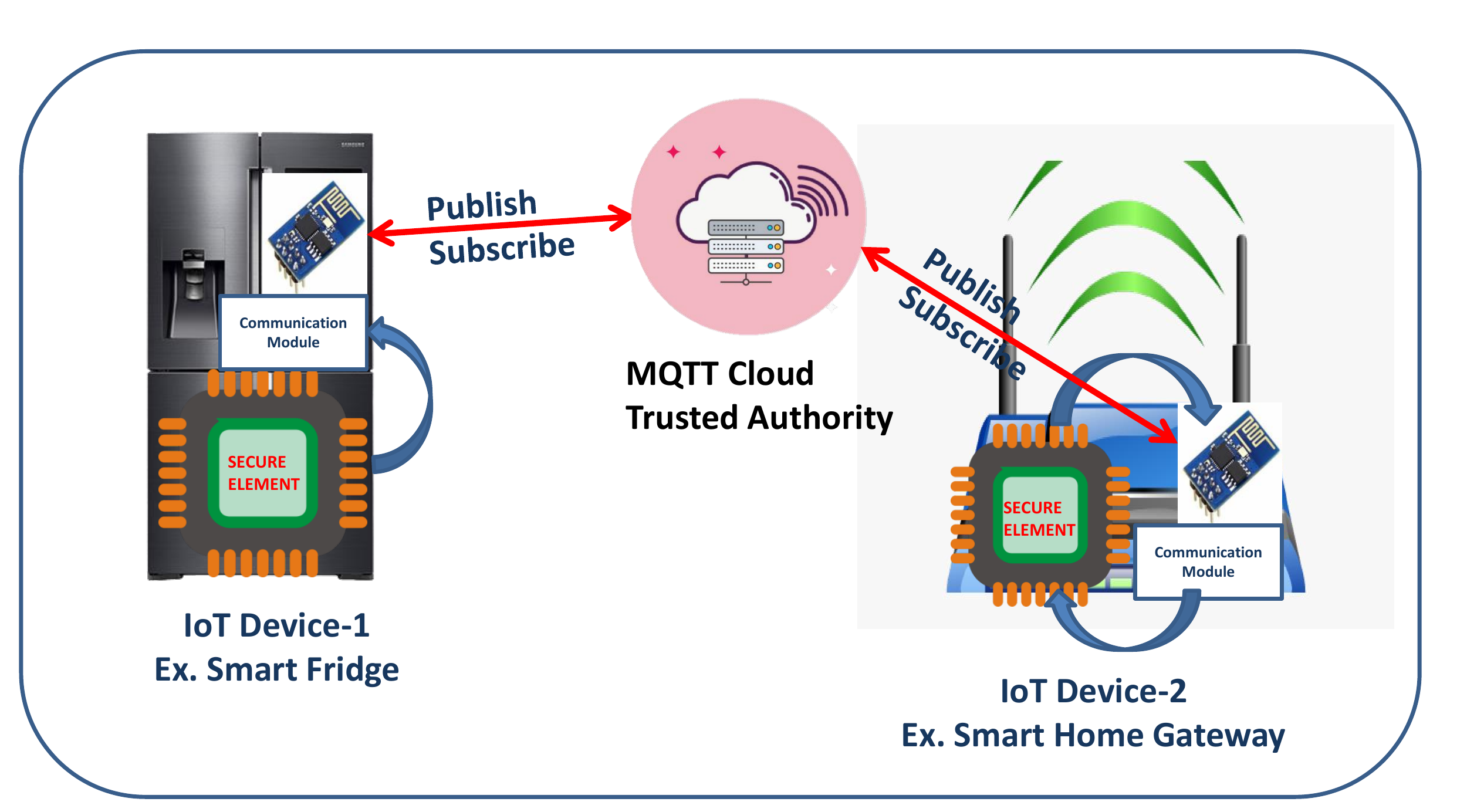}
\caption{Network model.} \label{fig2}
\end{figure}
\subsubsection{\textit{Threat Model}}
\noindent We adopted the Dolev-Yao channel based threat model for the proposed scheme. The attacker model or threat model for the proposed scheme is as follow:
\begin{itemize}
    \item Challenger $\mathcal{C}$ can read, access, modify, and store the communication over the open channel.
    \item Smart IoT devices, including gateway devices (in the presence of separate certificate authority or trusted authority), are not trusted devices. 
    \item Challenger $\mathcal{C}$ can capture the smart IoT device and extract the stored data over it.
    \item The TA is a trusted entity, and the polynomial-time challenger $\mathcal{C}$ can not compromise it. 
    \item Challenger $\mathcal{C}$ might receive the secrets of a TA in case of the system failure.
\end{itemize}
\subsubsection{Notations and Symbols}
\noindent Table \ref{Table:1} gives symbols and notations used for cryptanalysis and designing of the EBAKE-SE. 
\begin{table}[H]
    \caption{Notations and symbols}
    \centering
    \begin{tabular}{p{2cm}|p{5cm}} \hline
    \textbf{Symbols} & \textbf{Descriptions}  \\ \hline
    TA  & Trusted Authority \\
    $D_x$, $D_y$ & Xth and Yth Smart IoT devices \\
    $ID_x$, $ID_y$ & Identity of xth and yth smart IoT devices \\
    $TS_x$ / $T$ & Time-stamp \\
    $Pr_x$ & Private key of device $D_x$ \\
    $Pub_x$ & Public key of device $D_x$ \\
    $SK_{xy}$ & Generated session key \\
    Topic & MQTT topic \\
    $r_d$ & Random number \\
    $K_{dta}$ & 160 bit shared key \\
    $N_d$ & Random nonce \\
    $E_p(a,b)$ & Elliptic curve selected by TA \\
    \textit{Enc/Dec} & Encryption/Decryption \\
    $\bigoplus$ & Exclusive OR operation \\
    $P$ & Base point of the elliptic curve \\ 
    $D_{GWN}$ & Gateway device \\ \hline
    \end{tabular}
    \label{Table:1}
\end{table}

\section{Review of Das et al.'s scheme}
\label{review}
\noindent The scheme proposed by Das et al. \cite{Das2019} consists of four phases. (1) \textit{System setup phase} by TA (2) \textit{Device registration phase} by smart IoT device with TA (3) \textit{Device authentication phase} between two IoT smart device (4) \textit{Dynamic device addition phase} by TA. 

\noindent \textbf{\textit{(1) System setup phase:}} 

In this phase, the trusted authority (TA) decides finite field $\mathcal{F}_z$ and selects elliptic curve $E_z(a,b)$ (i.e. FIPS 186) over it. The TA also chooses base point $P$ of order x such that $x*P$ = $\mathcal{O}$ (infinity point). The TA generate pair of own private key and public key as a $(Pr_{TA},Pub_{TA})$ where $Pr_{TA}$ is randomly generated number and $Pub_{TA}$ = $Pr_{TA}*P$. Furthermore, the TA chooses one-way hash function h(.) (i.e. SHA1, MD5) for further processing and consistency between all the devices. Finally the TA publishes {$E_p(a,b)$, $P$, $p$, $Pub_{TA}$, h(.)} as a public parameters and stores $Pr_{TA}$, as a private parameter. Note that the TA is considered as a trusted entity \textcolor{black}{\cite{Sethia2018}}. 

\noindent \textbf{\textit{(2) Device registration phase:}} 

In this phase, the TA generates pair of \{$ID_x$, $Pr_x$, $A_x$, $c_x$, $Pub_x$, $E_p(a,b)$, $P$, $p$, $Pub_{TA}$, h(.)\} and then loads it over memory of the device $D_x$. Here $Pub_x$ = $Pr_x*P$, $A_x$ = $(Pr_x + l_x) * P$ where $l_x$ is a distinct random number for each device $D_x$ and  $c_x$ =$Pr_{TA}$ + $(Pr_x + l_x)$h($ID_x||A_x$). The pair of \{$ID_x$, $Pr_x$\} is generated by the TA for each device $D_x$.  

\noindent \textbf{\textit{(3) Device authentication phase:}}

In this phase, two smart IoT devices $D_x$ and $D_y$ performs authentication and set up a session key $SK_{xy}$ among each other. The summary of this phase is as follow:
\begin{enumerate}
    \item \textit{\textbf{$D_x \rightarrow D_y$ :}} The $D_X$ produces random $r_x$ and timestamp $TS_x$, computes $R_x$ = $r_x*P$, $z_x$ = $c_x$ + $h(A_x||c_x||R_x||Pub_x||TS_x)(r_x+Pr_x)$. The $D_x$ sends \textit{message 1} = \{$TS_x, ID_x, c_x, z_x, A_x, Pub_x, R_x$\} to other IoT device $D_y$.
    \item \textit{\textbf{$D_y \rightarrow D_x$ :}} The $D_y$ verifies timestamp, Verifies $U_y$ $\stackrel{?}{=}$ $c_x*P$ after computing $U_y$ = $Pub_{TA}+h(ID_x||A_x)A_x$, verifies $W_y$  $\stackrel{?}{=}$ $z_x*P$ after computing $W_y$ = $c_x*P$+$h(A_x||c_x||R_x||TS_x||Pub_x)(R_x+Pub_x)$. Next to these verification, the $D_y$ produces $TS_y$ and $r_y$ and computes $R_y$ = $r_y*P$, $z_y$ = $c_y$ + $h(A_y||c_y||R_y||Pub_y||TS_y)(r_y+Pr_y)$, $K_{xy}$ = $pr_y*Pub_x$, $B_{xy}$ = $r_y*R_x$, $SK_{xy}$ = $h(B_{xy}||K_{xy}||TS_y||TS_x||ID_x||ID_y$), $SKV_{xy}$ = $h(SK_{xy}||TS_y)$, sends \textit{message 2} = \{$ID_y, TS_y, A_y, c_y, z_y, SKV_{xy}, Pub_y, R_y$\} to device $D_x$.
    \item \textit{\textbf{$D_x \rightarrow D_y$ :}} The device $D_x$ verifies timestamp and verifies $U_x$ $\stackrel{?}{=}$ $c_y*P$ by computing $U_x$ = $Pub_{TA}+h(ID_y||A_y)A_y$. The device $D_x$ verifies $W_x$  $\stackrel{?}{=}$ $z_y*P$ by computing $W_x$ = $c_y*P$+$h(A_y||c_y||R_y||TS_y||Pub_y)(R_y+Pub_y)$, computes $K_{yx}'$ = $pr_x*Pub_y$, $B_{yx}'$ = $r_x*R_y$, $SK_{xy}'$ = $h(B_{yx}||K_{yx}||TS_x||TS_y||ID_y||ID_x$), verifies $SKV_{xy}$ $\stackrel{?}{=}$ $h(SK_{xy}'||TS_y)$. After this verification, the device $D_x$ produces timestamp $TS_x'$, computes $SKV_{yx}'$ = $h(SK_{yx}'||TS_x')$, generates \textit{message 3} = \{$SKV_{yx}', TS_x'$\} and sends it to the device $D_y$.
    \item \textit{\textbf{$D_y \rightarrow D_x$ :}} The device $D_y$ verifies timestamp and $SKV_{yx}*$ $\stackrel{?}{=}$ $SKV_{yx}'$ after computing $SKV_{yx}*$ = $h(SK_{yx}'||TS_x')$. After this verification both devices $D_x$ and $D_y$ agrees on the  session key $SK_{yx}'$ = $SK_{xy}$. 
\end{enumerate}

\noindent \textbf{\textit{(4) Dynamic device addition phase:}}
In this phase, the TA deploys new device or replaces device $D_x$ by $D_x'$. The TA selects $ID_x'$ and private key $Pr_x'$, computes public key $Pub_x'$ = $Pr_x'*P$, generates random number $l_x'$. The TA calculates $A_x'$ = $(Pr_x' + l_x') * P$, $c_x'$ =$Pr_{TA}$ + $(Pr_x' + l_x')$h($ID_x'||A_x'$) and stores \{$ID_x'$, $Pr_x'$, $A_x'$, $c_x'$, $Pub_x'$, $E_p(a,b)$, $P$, $p$, $Pub_{TA}$, h(.)\} in memory of the device $D_x'$.

\section{Cryptanalysis of Das et al.'s scheme}
\label{cryptanalysis}
\noindent In this section, we provide cryptanalysis for Das et al.'s and show that their scheme is vulnerable against attacks such as device impersonation, MITM, and DoS attack.
\subsection{Vulnerable Against Identity Theft attack/ Device Tracking Attack}
\noindent In the device authentication phase between device $D_x$ and $D_y$, 
\begin{itemize}
    \item Device $D_x$ sends \textit{message 1} = \{$TS_x, ID_x, A_x, c_x, z_x, Pub_x, R_x$\} to $D_y$ over an open channel.
    \item The \textit{message 1} contains identity $ID_x$ of the device $D_x$ in the plain text. The device $D_x$ does not protect it's identity inside \textit{message 1} though either hash or encryption. Thus, any challenger $\mathcal{C}$ can capture the $ID_x$ and use it for tracing the device $D_x$.
    \item Device $D_y$ sends \textit{message 2} = \{$ID_y, TS_y, A_y, c_y, z_y, SKV_{xy}, Pub_y, R_y$\} to $D_x$ over an open channel.
    \item The \textit{message 2} contains identity $ID_y$ of the device $D_y$ in the plain text. The device $D_y$ does not protect its identity inside \textit{message 2} though either hash or encryption. Thus, any challenger $\mathcal{C}$ can capture the $ID_y$ and use it for tracing the device $D_y$.
\end{itemize}
\subsection{Vulnerable Against Device Impersonation Attack / Device Capturing Attack / DoS}
\noindent In any IoT deployment, protecting the device from a \textit{physical device capturing} is a significant challenge. Authors in \cite{Das2019} do not provide any challenger limitations about the physical capturing of the smart devices. In the attacker model, Das et al. highlighted that the IoT device could be captured by the challenger $\mathcal{C}$. The challenger $\mathcal{C}$ can apply the \textit{power analysis attack} \cite{kocher1999} on any IoT device and can extract the stored information. Now let us examine Das et al.'s scheme against \textit{device impersonation attack}.
\begin{itemize}
    \item In the device registration phase, the TA loads \{$ID_x$, $Pr_x$, $A_x$, $c_x$, $Pub_x$, $E_p(a,b)$, $P$, $p$ ,$Pub_{TA}$, h(.)\} on device $D_x$. Now let us assume that the challenger $\mathcal{C}$ physically captures device $D_x$ and applies a power analysis attack on it. After performing successful power analysis attack, the challenger $\mathcal{C}$ already has \{$ID_x$, $Pr_x$, $A_x$, $c_x$, $Pub_x$, $E_p(a,b)$, $P$, $p$ ,$Pub_{TA}$, h(.)\}.
    \item Now, let us examine the first message generated by the device $D_x$. The device $D_x$ sends \textit{message 1} = \{$TS_x, ID_x, A_x, c_x, z_x, Pub_x, R_x$\} over an open channel. Now, the challenger $\mathcal{C}$ tries to generate a valid \textit{message 1*}.
    \item The challenger $\mathcal{C}$ already has \{$ID_x, A_x, c_x, z_x, Pub_x$\}. Now the challenger $\mathcal{C}$ generates random number $r_c$ from the public parameters of ECC and computes $R_c$ = $r_c$*P. Now, the challenger $\mathcal{C}$ also generates timestamp $TS_c$ and sends message \textit{message 1*} = \{$TS_c, ID_x, A_x, c_x, z_x, Pub_x, R_c$\} to device $D_y$.
    \item Now the device $D_y$ verifies timestamps, and computes $U_y$ = $Pub_{TA}+h(ID_x||A_x)A_x$, The device $D_y$ successfully verifies $W_y$  $\stackrel{?}{=}$ $z_x*P$ after computing $W_y$ = $c_x*P$+$h(A_x||c_x||R_c||TS_c||Pub_x)(R_c+Pub_x)$. Thus the challenger can also generate \textit{message 1*} that leads to valid \textit{device impersonation}.  
    \item In the proposed scheme of Das et al., the device $D_x$ or the device $D_y$ does not block fake devices even after numerous failed attempts from the sender. Thus, this can easily drain the receiving device's battery and may lead to power failure. Therefore, we can say that any malicious attacker can send fake requests and leads the system toward \textit{DoS}.
\end{itemize}
\subsection{Vulnerable Against MITM Attack / Fake Session Key Setup}
\noindent The scheme of Das et al. is also vulnerable against MITM attack. In the scheme proposed by Das et al.,
\begin{itemize}
    \item Let us assume that there is a malicious intruder $\mathcal{C}$ eavesdrops public message  \textit{message 1} = \{$TS_x, ID_x, A_x, c_x, z_x, Pub_x, R_x$\} and  \textit{message 2} = \{$ID_y, TS_y, A_y, c_y, z_y, SKV_{xy}, Pub_y, R_y$\}. Now let us assume that $\mathcal{C}$ computes $B_{cj}$ = $r_c$*$R_x$ and $K_{cj}$ = $x_c$*$Q_x$ and generates $SK_{xc}**$ = $h(B_{cj}||K_{cj}||TS_y||TS_x||ID_x||ID_y$), $SKV_{xc}**$ = $h(SK_{xc}||TS_y)$ and forward to device $D_x$. 
    \item We must note here that the challenger $\mathcal{C}$ only replaces $SKV_{xy}$ by $SKV_{xc}**$ and sends the remaining \textit{message 2} as it is. Thus the device $D_x$ can not identify that the received message is from the challenger $\mathcal{C}$, not from the valid device $D_y$. The device $D_x$ uses $Pub_y$ and $R_y$ from the \textit{message 2} (not from the previous knowledge) for the computation of the $B_{ij}'$ and $K_{ij}'$. Thus, unknowingly, the device $D_x$ establishes the session key with the challenger $\mathcal{C}$.   
\end{itemize}
\section{Proposed scheme : EBAKE-SE}
\label{proposed}
\noindent The \textit{secure element (SE)} is a tamper-resistant microprocessor chip that stores secret data for the tiny devices and securely runs their applications. The \textit{secure element} is embedded with the IoT devices in such a way that the logical tempering on it becomes an impossible task and the physical tempering of a \textit{secure element} destroys the functioning of the device \cite{secureelement}. In the proposed setup, we consider that both the IoT devices are embedded with the \textit{secure element} on it. Fig. \ref{fig2} shows the communication model for the proposed EBAKE-SE. In EBAKE-SE, we consider the MQTT Cloud server as a resource-capable, trusted authority (TA) that runs the MQTT broker module. We highlight more detail about the MQTT protocol in section \ref{implelment}. In this section, we provide improvements for the scheme proposed by Das et al. \cite{Das2019}. In the proposed EBAKE-SE, there are two major phases. In the first phase, the TA initializes the system, generates necessary parameters, and stores those parameters in the SE of the smart IoT devices. In the second phase, two IoT devices perform mutual authentication via TA and generate a one-time session key ($SK_{xy}$) for further secure communications. In this phase, the TA also allocates a temporary (for a session) MQTT topic on which these devices perform encrypted communication. In the proposed EBAKE-SE, each smart IoT device has two connected modules. The first module is the SE module, which runs cryptographic operations. The second module is a wifi module (we used the esp8266 module for implementation), which connects the device with the internet for communication with TA using the MQTT protocol. The proposed EBAKE-SE overcomes the limitations of the analyzed scheme and introduces some novel features compared to other existing schemes proposed for a similar environment. 
\subsection{System Initialization Phase}
\noindent In this phase, the TA generates credentials for self and smart IoT devices and loads those credentials over the SE of the IoT device. The TA performs initialize phase in a secure environment as follows: The TA selects a base point $P$ for the curve $E_p(a,b)$. The TA generates unique identity for each $x^{th}$ device as $ID_d^x$, generates random number for each $x^{th}$ device as a $r_d^x$, and generates shared secret $K_{dta}$ between the device $D_x$, other IoT devices and itself. TA updates $K_{dta}$ periodically. TA computes device parameter $DP_1^x$ : \textit{hash \big \langle $ID_d^x$, $r_d^x$, $K_{dta}$ \big \rangle}. TA computes public parameter $Q_d^x$ : $r_d^x$*P for each $x^{th}$ smart IoT device. TA loads pair \textit{\big \langle $ID_d^x$, $r_d^x$, $K_{dta}$, $DP_1^x$ \big \rangle} on SE of the device $D_x$. TA also loads pair \textit{\big \langle $ID_d^x$, $DP_1^x$, $K_{dta}$, $Q_d^x$ \big \rangle} on its own secret memory.  
\subsection{Mutual Authentication Phase}
\noindent In the IoT setup, each party must trust the other. In this phase, initially we perform the mutual authentication between devices \textit{\big \langle $D_x$, $TA$ \big \rangle}, \textit{\big \langle $D_y$, $TA$ \big \rangle}, and \textit{\big \langle $D_x$, $D_y$ \big \rangle}. This is followed by a secure session key generation between devices $D_x$ and $D_y$ as a $SK_{xy}$ and topic allocation by $TA$. The system performs mutual authentication as follows: 

The device $D_x$ generates temporary id $ID_T^x$: \textit{\big \langle $W^x$, $Y^x$, $Z^x$ \big \rangle} as follow: 

\textit{\textbf{Step-1 :}} The device $D_x$ generates random nonce $N_d^x$ and computes $W^x$: \textit{Enc \big \langle  ($K_{dta}, (ID_d^x, r_d^x) $\big \rangle}, $Y^x$: \textit{xor \big \langle  ($DP_1^x$, $Q_d^y$) \big \rangle}, $Z^x$: \textit{Enc \big \langle  ($Q_d^y, (Q_d^x, ID_x N_d^x, T_1) $\big \rangle}, $P_d^x$: \textit{hash \big \langle $DP_1^x$, $N_d^x$, $T_1$ \big \rangle}. The device $D_x$ publishes \big \langle $ID_T^x$, $P_d^x$, $T_1$ \big \rangle to TA. 

\textit{\textbf{Step-2 :}} The TA receives $ID_T^x$ and performs as follow: TA first verifies timestamp and then verifies identity of the sending device as : The TA verifies $\Delta T \stackrel{?}{\leq} T_1^{*}$ - $T_1$, retrieves pair \textit{\big \langle ($ID_d^x*, r_d^x*) $\big \rangle} by: \textit{Dec \big \langle  $K_{dta}$, ($W^x$) \big \rangle}. The TA computes $DP_1^x$*: \textit{hash \big \langle $ID_d^x*$, $r_d^x*$, $K_{dta}*$ \big \rangle}, computes $P_d^x$*:\textit{hash \big \langle $DP_1^x*$,$T_1$ \big \rangle} and verifies $P_d^x$* $\stackrel{?}{\leq}$ $P_d^x$. After three unsuccessful verification from the same device, the TA blocks the device for a day. Now, the TA retrieves  $Q_d^y$* : \textit{xor \big \langle  ($DP_1^x$*, $Y^x$)}. The TA identifies $D_y$, computes $P_d^y$: \textit{hash \big \langle $DP_1^y$, $T_2$}, and publishes \big \langle $Z^x$, $P_d^y$, $T_2$ \big \rangle to $D_y$.  

\textit{\textbf{Step-3 :}} The $D_y$ receives pair \big \langle $Z^x$, $P_d^y$, $T_2$ \big \rangle. The $D_y$ verifies $\Delta T \stackrel{?}{\leq} T_2^{*}$ - $T_2$  and retrieves \big \langle  $Q_d^x$, $ID_d^x$, $N_d^x$, $T_1$  \big \rangle by \textit{Dec \big \langle ($r_d^y, (Z^x) $\big \rangle}. The device $D_y$ verifies $P_d^y$ $\stackrel{?}{\leq}$ $P_d^y*$: \textit{hash \big \langle $DP_1^y$, $T_2$}. By this verification, the device $D_y$ authenticates TA. After three unsuccessful authentications, the device $D_y$ blocks $TA$ for a day by considering it a DoS attack from the malicious insider. Now the device $D_y$ generates a nonce $N_d^y$, computes $Z^y$: \textit{Enc \big \langle  ($Q_d^y, (ID_y, N_d^y, T_2) $\big \rangle}, computes $P_d^{TA}$: \textit{hash \big \langle $DP_1^x$, $ID_d^x$, $ID_d^y$, $T_3$, $ID_d^y$} and publishes pair \big \langle $Z^y$, $P_d^{TA}$, $T_3$ \big \rangle to $TA$. The device $D_y$ computes one-time secure session key for the device $D_x$ as $SK_{xy}$ : \textit{hash \big \langle $ ID_y, N_d^y, T_1, ID_x, N_d^x, T_2, K_{dta}) $ \big \rangle}. 

\textit{\textbf{Step-4 :}} The TA receives data from the device $D_y$ and verifies $\Delta T \stackrel{?}{\leq} T_3^{*}$ - $T_3$. Now the TA also verifies $P_d^TA$ $\stackrel{?}{\leq}$ $P_d^TA$: \textit{hash \big \langle $DP_1^x$, $ID_d^x$, $ID_d^y$, $T_3$, $ID_d^y$}. After three unsuccessful verification from the same device, the TA blocks the device for a day. Now the TA computes $P_d^{xx}$: \textit{hash \big \langle $DP_1^x$, $Z^y$, $T_4$ \big \rangle}, and publishes pair \big \langle $Z^y$,$T_4$ \big \rangle along with MQTT topic \textit{T} to device $D_x$. The TA shares the same MQTT topic (\textit{T}) with the device $D_y$. 

\textit{\textbf{Step-5 :}} The device $D_x$ verifies $\Delta T \stackrel{?}{\leq} T_4^{*}$ - $T_4$, verifies $P_d^{xx}$ $\stackrel{?}{\leq}$ $P_d^{xx}*$: \textit{hash \big \langle $DP_1^x$, $Z^y$, $T_4$ \big \rangle}. By verifying the device, $D_x$ authenticates both TA and the device $D_y$. After three unsuccessful authentications, the device $D_x$ blocks stops communication with TA for a day by considering it a DDoS attack from the malicious insider. The device $D_x$ retrieves pair \big \langle  ($Q_d^y, (ID_y, N_d^y, T_2) $ \big \rangle by \textit{Dec \big \langle  ($r_d^x, (Z^y) $ \big \rangle},  and computes one-time secure session key for the device $D_y$ as $SK_{xy}$ : \textit{hash \big \langle $ ID_x, N_d^x, T_1, ID_y, N_d^y, T_2, K_{dta} $ \big \rangle}. The device $D_x$ and $D_y$ starts $SK_{xy}$ encrypted  communication over a given topic \textit{T}.  

\vspace{0.1in}
\noindent Thus, after completion of this phase, both the devices have a pair of \textit{\big \langle $SK_{xy}$, \textit{T} \big \rangle}. We like to observe that even though we perform mutual authentication via $TA$, the $TA$ can not compute the final session key $SK_{xy}$ due to a lack of awareness about the random numbers ($r_d^x$, $r_dy$) and the random nonces ($N_d^y$, $N_d^y$). The verification parameters ($P_d^{x}$, $P_d^{y}$, $P_d^{XX}$, $P_d^{TA}$) provide strength to the proposed work. The use of timestamps prevents an intruder from performing a replay-type attack. In the proposed EBAKE-SE, to protect a device from the DoS and DDoS type attacks, we block malicious devices for a day if the receiver could not authenticate it after three verification. The novelty in the proposed scheme lies with the use of the tamper-resistant SE on each IoT device.

\section{Informal security analysis}
\label{informal}
\noindent In this section, we show that the proposed EBAKE-SE achieves desired security goals and resists all well-known attacks with excellent cryptography functions. Table \ref{securitygoals} highlights security features based comparison between the proposed scheme and other existing schemes. 
\subsection{Achieves Security Against Traditional and Non-traditional Attacks}
This subsection provides proof of the "informal security" for the proposed EBAKE-SE. 
\begin{enumerate}[F1.]
\item \textbf\textit{{EBAKE-SE is secure against a reply attack:}}  We involve random numbers and timestamps in all the exchanged messages during the mutual authentication phase of the proposed EBAKE-SE. Use of the random numbers \{$N_d^x$, $N_d^y$\}, and timestamps \{$T_1$, $T_2$, $T_3$, $T_4$\} guarantees freshness of the communicated messages. As a result, the proposed EBAKE-SE is free from replay attacks.

\item \textcolor{black}{EBAKE-SE is secure against a MITM Attack:} Suppose a challenger $\mathcal{C}$ expropriate the valid authentication messages and tries to modify these messages to another valid authentication message. It is "computationally infeasible challenge" for a challenger $\mathcal{C}$ to generate valid authentication message \{$ID_T^x$, $P_d^x$, $T_1$\} due to unawareness about the shared secret $K_{dta}$ stored in SE and original random nonce $N_d^x$. In like manner, $\mathcal{C}$ can not also generate other valid authentication messages. This obliques that the proposed EBAKE-SE achieves protection from the "Man-In-The-Middle" attack.   

\item EBAKE-SE is secure against an impersonation attack: In impersonation attack, the challenger $\mathcal{C}$ tries to create a valid authentication \textit{message} \{$ID_T^x$, $P_d^x$, $T_1$\} to pretend as a valid device $D_x$. The challenger $\mathcal{C}$ must require secret parameters such as \{$K_{dta}, ID_d^x, r_d^x$\} to generate  \textit{message}. These secret parameters are stored in SE, and it is impossible for the challenger $\mathcal{C}$ to obtain these values. Thus, eavesdropping of \textit{message} will not allow a challenger $\mathcal{C}$ to generate similar \textit{message*} to impersonate as a device $D_x$. Similarly, $\mathcal{C}$ can not also pretend as a device $D_y$. Hence, the proposed EBAKE-SE is immune enough against an impersonation attack.

\item \textcolor{black}{EBAKE-SE preserves anonymity and traceability: }Suppose a challenger $\mathcal{C}$ captures messages \{$ID_T^x$, $P_d^x$, $T_1$\}, \{$Z^x$, $P_d^y$,$T_2$\}, \{$Z^y$, $P_d^{TA}$, $T_3$\}, \{$Z^y$,$T_4$\} and tries to trace the devices $D_x$ and $D_y$. To trace the devices, challenger $\mathcal{C}$ must require either static messages or public identity. In the proposed EBAKE-SE, each message is an output of the random values, and none of the public messages contains the identity of either device in the plain text. Therefore the proposed EBAKE-SE achieves anonymity and traceability.   
\begin{table}
    \caption{Security Features and Goals}
    \scalebox{.90}{\begin{tabular}{p{1.2cm}p{0.3cm}p{0.3cm}p{0.3cm}p{0.3cm}p{0.3cm}p{0.3cm}p{0.3cm}p{0.3cm}p{0.3cm}} \hline
    \textbf{Scheme} & \textbf{$F_1$} & \textbf{$F_2$} & \textbf{$F_3$} & \textbf{$F_4$} & \textbf{$F_5$} & \textbf{$F_6$} & \textbf{$F_7$} & \textbf{$F_8$} & \textbf{$F_9$} \\ \hline
    \textcolor{black}{Proposed} & \checkmark & \checkmark & \checkmark & \checkmark & \checkmark & \checkmark & \checkmark  & \checkmark & \checkmark \\
    \textcolor{black}{\cite{Das2019}} & \checkmark & \xmark & \xmark & \checkmark & \checkmark & \xmark & \checkmark  & \checkmark & \checkmark \\ 
    \textcolor{black}{\cite{challa2017secure}} & \checkmark & \checkmark & \checkmark & \checkmark & \checkmark & \checkmark & \xmark  & \xmark  & \xmark  \\
    \textcolor{black}{\cite{li2017robust}} & \checkmark & \xmark & \checkmark & \checkmark & \checkmark & \xmark & \checkmark  & \checkmark & \checkmark \\
    \textcolor{black}{\cite{wu2018secure}} & \checkmark & \checkmark & \checkmark & \checkmark & \xmark & \checkmark & \checkmark  & \xmark & \xmark \\ 
    \textcolor{black}{\cite{porambage2015group}} & \checkmark & \xmark & \checkmark & \checkmark & \xmark & \checkmark & \checkmark  & \checkmark & \checkmark \\ 
     \textcolor{black}{\cite{li2019}} & \xmark & \xmark & \checkmark & \checkmark & \checkmark & \checkmark & \xmark  & \checkmark & \xmark \\ \hline
    \end{tabular}}
    \label{securitygoals}
\end{table}
\item EBAKE-SE is secure against a secret leakage attack: 
In the proposed scheme, we use long term secrets \{$K_{dta}$,  $r_d^x$\} and session-specific temporary nonces \{$N_d^x$, $N_d^y$\}. The session key is computed as a $SK_{xy}$ : \textit{hash \{$ID_x, N_d^x, T_1, ID_x, N_d^x, T_2, K_{dta}\} $}. Now let us assume challenger $\mathcal{C}$ reveals pair \{$K_{dta}$,  $r_d^x$\} then also he/she can not compute the session key because of non availability of \{$ID_x, N_d^x, ID_x, N_d^x\} $ \big \rangle. Similarly, exposure of any information does not allow a challenger $\mathcal{C}$ to validate a key. Hence, we derive that EBAKE-SE is secured against a secret leakage attack.

\item EBAKE-SE is secure against insider attack:
Suppose a malicious administrator on TA tries to compute the session key using available data. The malicious administrator retrieves stored parameters \{$ID_d^x$, $DP_1^x$, $K_{dta}$, $Q_d^x$ \} as well as receives public messages \{$ID_T^x$, $P_d^x$, $T_1$\}, \{$Z^x$, $P_d^y$,$T_2$\}, \{$Z^y$, $P_d^{TA}$, $T_3$\}, \{$Z^y$,$T_4$\}. The malicious administrator does not get random nonces \{$N_d^x$, $N_d^y$\} necessary for session key computations. In the proposed EBAKE-SE, TA does not store \{$r_d^x$, $r_d^y$\}. Hence, the proposed EBAKE-SE is free from malicious insider attacks.

\item EBAKE-SE achieves session key agreement: In the proposed EBAKE-SE, mutual authentication between the smart devices and TA is achieved by following verifications: $P_d^x$* $\stackrel{?}{\leq}$ $P_d^x$ (By TA for $D_x$), $P_d^y$ $\stackrel{?}{\leq}$ $P_d^y*$ (By $D_y$ for TA), $P_d^TA$ $\stackrel{?}{\leq}$ $P_d^TA$ (By TA for $D_y$) and $P_d^{xx}$ $\stackrel{?}{\leq}$ $P_d^{xx}*$ (By $D_x$ for TA and $D_y$). The session key computation involves insider parameters from these validations $SK_{xy}$ : \textit{hash \big \langle $ ID_y, N_d^y, T_1, ID_x, N_d^x, T_2, K_{dta}) $ \big \rangle}. Therefore, we derive that the proposed EBAKE-SE achieves session key agreement. 
\item EBAKE-SE is secure against perfect forward secrecy : Suppose a challenger $\mathcal{C}$ obtains shared secret credentials $K_{dta}$, furthermore, the challenger intercepts the messages \{$ID_T^x$, $P_d^x$, $T_1$\}, \{$Z^x$, $P_d^y$,$T_2$\}, \{$Z^y$, $P_d^{TA}$, $T_3$\}, \{$Z^y$,$T_4$\} communicated between the smart devices via TA. To obtain the previous session key, the challenger $\mathcal{C}$ must compute $SK_{xy}'$ = \textit{hash \big \langle $ ID_y', N_d^y, T_1', ID_x', N_d^x, T_2', K_{dta}' $ \big \rangle}. Even though, if adversary also obtains identity of devices somehow, he/she must extract past random nonces \{$N_d^x$, $N_d^y$\} protected though encryption. Hence, the proposed EBAKE-SE provides perfect forward secrecy.  
\end{enumerate}
\section{Formal security analysis using ROR}
\label{formal}
\noindent In this section, we provide a formal security model for the session key ($SK_{xy}$) derived as an outcome of EBAKE-SE. A random oracle-based Real-Or-Random (ROR) model is used for the formal security modelling of the proposed EBAKE-SE. \textcolor{black}{Recently, many researcher in \cite{jhaveri2021fault} \cite{Das2019} adopted ROR model for their security validations.} ROR follows the principle of "indistinguishability" between a real session key and a random number. We first instigate the ROR security model and then provide the security proof for the proposed EBBAC-SE under the instigated model.
\subsection{Security Model}
We define a security model of the proposed EBAKE-SE using a game between a probabilistic polynomial time (PPT) challenger $\mathcal{C}$ and a responder $\mathcal{R}$. In this game, the challenger $\mathcal{C}$ loads oracle queries, and the responder $\mathcal{R}$ responds to these queries. Let us consider three participants (smart IoT device $D_x$, smart IoT device $D_y$, and trusted authority $TA$) in the proposed protocol $\mathcal{P}$.    

\textbf{Responder Model:} Let us define oracle instances for responders: $\mathcal{O}^{l}_{TA}$, $\mathcal{O}^{m}_{D_x}$, $\mathcal{O}^{n}_{D_y}$, are oracles of $l$, $m$ and $n$ for the $TA$ instances, device $D_x$ and the device $D_y$ respectively. These participants are called fresh if they do not reveal the original session key as a response to the $\mathcal{R}$ query by $\mathcal{C}$. These participants are called partners if they share a common session-id $S_{id}$ transcript of all the communicated messages. These participants are commonly considered as $\mathcal{D_l}$ if it is not necessary to represent them separately.  

\textbf{Challenger Model:} We design a challenger $\mathcal{C}$ using the famous \textit{Dolev-Yao model}. The challenger $\mathcal {} $ can perform active and passive attacks over the Dolev-Yao channel. Following random oracle, queries define capabilities for a PPT challenger $\mathcal{C}$.

\textit{Execute Query:} \textbf{$\mathcal E(\mathcal{O}^{l}_{TA}, \mathcal{O}^{m}_{D_x}, \mathcal{O}^{n}_{D_y})$} query provides all communicated messages over open channel between all participants. This query is a passive attack over the proposed protocol $\mathcal{P}$.

\textit{Reveal Query:} \textbf{$\mathcal R(\mathcal{O}^{m}_{D_x}):$)} query responds session key $SK$ to the challenger $\mathcal{C}$ if the responder $\mathcal{R}$ accepts it.

\textit{Hash Query:} \textbf{$\mathcal H(m_x)$} query responds random $r_x$ and stores it in a list $\mathcal{L_x}$ defined with a null value by the responder $\mathcal{R}$.

\textit{Send Query:} \textbf{$\mathcal {S}(\mathcal{O}^{m} _{D_x}, m_x)$} query is presented as an active intrusion over proposed protocol $\mathcal{P}$. The challenger $\mathcal{C}$ sends message $m_x$ to the responder $\mathcal{R}$ and gets the reply from $\mathcal{R}$ according to the specifications of the message $m_x$.

\textit{Test Query:} \textbf{$\mathcal T(\mathcal{O}^{m}_{D_x}):$)} query responds either true session key or a equal size random element. The responder $\mathcal{R}$ randomly selects a bit \textit{u}. If $\mathcal{R}$ randomly selects \textit{u} = 1 then it returns original session key else (means \textit{u} = 0) it returns random element with equal bit length of $SK$ to the challenger $\mathcal{C}$. 

\textit{Corrupt Query:} \textbf{$\mathcal CR(\mathcal{O}^{m}_{D_x}):$)} query responds data stored inside memory of the responder $\mathcal{R}$ to challenger $\mathcal{C}$. \textcolor{black}{Through this query challenge can get any data store inside non-secure memory of IoT devices}

The challenger tries all these queries for finite times, and after executing these queries, $\mathcal{C}$ guesses the value of bit \textit{u} as \textit{u'}. Let \textit{$Adv_{\mathcal{P}}$} represent winning event (retrieves original session key) for the challenger $\mathcal{C}$ and $\mathcal{SUC}$ represent the success position for $\mathcal{C}$. We can define the challenger $\mathcal{C}$'s advantage of breaking the proposed EBAKE-SE as:
\begin{equation}
    \label{eq3}
    \textit{$Adv_{\mathcal{P}}$}(\mathcal{C}) = 2*Pr[\mathcal{SUC}] - 1,
\end{equation}    
\hspace{4cm}    OR 
\begin{equation}    
    \label{eq4}
    \textit{$Adv_{\mathcal{P}}$}(\mathcal{C}) = 2*Pr[u'= u] - 1,
\end{equation}
Let $q_s$ represent the number of send queries, $l_h$ represent the hash length, $l_r$ represent the length of random elements, $q_h$ represent the number of the hash query, and $q_e$ represent the number of executing query, we can give formal security proof for the proposed EBAKE-SE as follow:  
\subsection{Formal Security Proof}
\begin{theorem}
We consider the cyclic group $\mathcal{G}$ of order \textit{n} to define an elliptic curve \textit{E} over finite field \textit{$F_p$}. We define the finite time $t_c$ during it the challenger $\mathcal{C}$ tries $q_h$, $q_e$ and $q_s$ to break the proposed protocol $\mathcal{P}$. We can define security for the proposed $\mathcal{P}$ against oracle queries loaded by the challenger $\mathcal{C}$ as,
\begin{equation}
\begin{split}
   \label{eq5}
    \textit{$Adv_{\mathcal{P}}$}(\mathcal{C}) \leq \frac{q^2_{h}}{2^{l_h+1}} + \frac{(q_{s}+q_{e})^2}{2^{l_s+1}} + (4*q_e+2*q*s) \\ \textit{$Adv_{\mathcal{C}^{ECDH}}$} (t^*) + max(q_s, (\frac{1}{2^l}, \rho_{fp})),
\end{split}
\end{equation}
For any given \textit{xP} and \textit{yP}, the \textit{$Adv_{\mathcal{C}^{ECDH}}$}($t^*$) represent the polynomial time ($t^*$) probability for the challenger $C$ to break the elliptic curve diffie-hellman problem and compute the valid \textit{xyP} value.  
\end{theorem}
\begin{proof}
We define four identical security games \{$Gm_0$, $Gm_1$, $Gm_2$, $Gm_3$\} which proves the proposed protocol $\mathcal{P}$ is secured against P.P.T. challenger $\mathcal{C}$ under ROR model and \textit{$Adv_{\mathcal{P}}$}($\mathcal{C}$) is negligible under random oracle game. Let $Suc_i$ define the probability of correctly guessing the value of bit \textit{u} by the challenger $\mathcal{C}$ for the game $Gm_i$ during the challenge session. 

\vspace{0.1in}
\textit{Game $Gm_0$:} The game $Gm_0$ is an identical game to real protocol. If the challenger $\mathcal{C}$ takes more time than a threshold $t^*$ or does not respond to the game, then the arbitrary value for the bit \textit{u} will be selected. Thus, it is apparent that,
\begin{equation}
    \label{eq6}
    \textit{$Adv_{\mathcal{P}}$}(\mathcal{C}) = 2*Pr[Suc_0] - 1,
\end{equation}

\vspace{0.1in}
\textit{Game $Gm_1$:}
In this game, the challenger $\mathcal{C}$ performs executive query \textbf{$\mathcal{E}$} to eavesdrop the communication between devices ($D_x$ and $D_y$) and the trusted authority (TA). 
\begin{itemize}
    \item \textbf{$\mathcal E(D_x, TA):$} is loaded for capturing the communication between the device $D_x$ and TA. 
    \item \textcolor{black}{\textbf{$\mathcal E(TA, D_y):$} is loaded for receiving the communication between the device $D_y$ and TA.}
\end{itemize}
The challenger $\mathcal{C}$ stores all the messages extracted from the above queries and try to compute the session key $SK_{xy}$. If the challenger $\mathcal{C}$ could compute the session key then challenger $\mathcal{C}$ captures the game $Gm_1$, otherwise it is considered that Pr[$Suc_1$] = Pr[$Suc_0$]. In the proposed scheme, we compute the final session key $SK_{xy}$ : \textit{hash \big \langle $ ID_x, N_d^x, T_1, ID_y, N_d^y, T_2, K_{dta}) $ \big \rangle} using the random nonces and the nonpublic identities with a shared secret, Hence,  
\begin{equation}
    \label{eq7}
     Pr[Suc_1] = Pr[Suc_0],
\end{equation}
In the proposed scheme, we compute the final session key $SK_{xy}$ : \textit{hash \big \langle $ ID_x, N_d^x, T_1, ID_y, N_d^y, T_2, K_{dta} $ \big \rangle} using the random nonces and the nonpublic identities with shared secret, hence it is infeasible for the challenger $\mathcal{C}$ to compute the session key using captured information that is identical to the game $Gm_0$. Therefore, the equation \ref{eq7} holds true. 

\textit{Game : $Gm_2$}
In this game, the challenger $\mathcal{C}$ performs \textbf{$\mathcal{H}$} and \textbf{$\mathcal{S}$} query to communicate with the devices ($D_x$ and $D_y$) and the trusted authority (TA). In this game, the challenger $\mathcal{C}$ tries to create a collision for the establishment of a fake trust. We can define collision probability of hash function using the birthday paradox at most $\frac{q^2_{h}}{2^{l_h+1}}$. Each communicated message in the proposed protocol $\mathcal{P}$ built up using the random nonces ($N_d^y$, $N_d^x$), random numbers ($r_d^x$,$r_d^xy$) and timestamps ($T_i$). The collision probability for these values is at most $\frac{(q_{s}+q_{e})^2}{2^{l_s+1}}$. Thus the game $Gm_2$ and the game $Gm_1$ are identical games till the collision arises, hence,
\begin{equation}
    \label{eq8}
    Pr[Suc_2] - Pr[Suc_1] \leq \frac{q^2_{h}}{2^{l_h+1}} + \frac{(q_{s}+q_{e})^2}{2^{l_s+1}},
\end{equation}
\textit{Game : $Gm_3$}
In this game, the challenger $\mathcal{C}$ performs corrupt query \textbf{$\mathcal CR(\mathcal{O}^{m}_{D_x}):$)} and send query \textbf{$\mathcal{S}$} or an execute query \textbf{$\mathcal{E}$} with the random oracles. The challenger also tries to solve the ECDH problem of the ECC. Let us consider that the challenger $C$ tries following queries,
\begin{itemize}
    \item Using \textbf{$\mathcal CR(\mathcal{O}^{m}_{D_x})$):} query, the challenger retrieves \textit{\big \langle $ID_d^x$, $r_d^x$, $K_{dta}$, $DP_1^x$ \big \rangle} 
    \item \textcolor{black}{Using \textbf{$\mathcal CR(\mathcal{O}^{m}_{D_y})$):} query, the challenger retrieves \textit{\big \langle $ID_d^x$, $DP_1^x$, $K_{dta}$, $Q_d^x$ \big \rangle}}
    \item Using \item \textbf{$\mathcal E(D_x, TA)$:} query, the challenger retrieves \big \langle $ID_T^x$, $P_d^x$, $T_1$ \big  \rangle, \big \langle $Z^y$,$T_4$ \big \rangle.
    \item \textcolor{black}{Using \textbf{$\mathcal E(TA, D_y)$:} query, the challenger retrieves \big \langle $Z^x$, $P_d^y$, $T_2$ \big \rangle, \big \langle $Z^y$, $P_d^{TA}$, $T_3$ \big \rangle.} 
\end{itemize}
After performing the following queries for a finite time, the challenger tries to decrypt the data encrypted by the public keys \{$Q_d^x$, $Q_d^y$\}. These public keys are computed as $Q_d^x$ = $r_d^x$*P and $Q_d^y$ = $r_d^y$*P. For given \{$Q_d^x$, $Q_d^y$\} and P, it is computationally infeasible to find the value of \{$r_d^x$, $r_d^y$\}. The probability of solving the ECDH problem is at most (4*$q_e$+2*q*s)\textit{$Adv_{\mathcal{C}^{ECDH}}$}($t^*$). The probability of guessing the correct random nonces ($N_d^y$, $N_d^x$) after performing the \textbf{$\mathcal CR(\mathcal{O}^{m}_{D_x})$)} and \textbf{$\mathcal CR(\mathcal{O}^{m}_{D_y})$)} is at most max($q_s$, ($\frac{1}{2^l}$, $\rho_{fp}$)). It is infeasible for the challenger $\mathcal {} $ to solve the ECDH problem and guess the correct random numbers simultaneously in polynomial time. Hence, the game $Gm_3$ is identical to the game $Gm_2$. Thus we have,  
\begin{equation}
\begin{split}
    \label{eq9}
    Pr[Suc_3] - Pr[Suc_2] \leq (4*q_e+2*q*s)\textit{$Adv_{\mathcal{C}^{ECDH}}$}(t^*) + \\ max(q_s, (\frac{1}{2^l}, \rho_{fp})),
\end{split}
\end{equation}
\end{proof}
Now the challenger $\mathcal{C}$ tries to guess the bit \textit{u'} and the probability of correct guess is at most $\frac{1}{2}$. Thus, from equations \ref{eq8} and \ref{eq9}, we can derive,
\begin{equation}
\begin{split}
 \label{eq10}
\textcolor{black}{\textit{$Adv_{\mathcal{P}}$}(\mathcal{C}) \leq  \frac{q^2_{h}}{2^{l_h+1}}} + \frac{(q_{s}+q_{e})^2}{2^{l_s+1}} + (4*q_e+2*q*s) \\ \textit{$Adv_{\mathcal{C}^{ECDH}}$} (t^*) + max(q_s, (\frac{1}{2^l}, \rho_{fp})),
\end{split}
\end{equation}
\section{Implementation using MQTT}
\label{implelment}
\noindent The Message queuing telemetry transport (MQTT) protocol is a widely adopted publish-subscribe-based, lightweight application layer protocol for communicating in the IoT-based environment. In the MQTT protocol, there are three entities, (1) The publisher (who publishes the data), (2) The subscriber (who receives the data), and the broker (who integrates and forwards the data). To implement the proposed protocol, we used Raspberry Pi 3 Model B (with Quad Core 1.2GHz Broadcom BCM2837 64bit CPU and 1GB RAM) as a sensing device and the laptop device installed with the mosquitto broker on it. We can also utilize global brokers (such as AWS and hivemq). \textcolor{black}{For sniffing purposes, we utilized laptop devices and installed the mosquitto broker and Wireshark tool over it.} We used \textit{Paho} library that provides MQTT client services. We implemented the proposed EBAKE-SE with the fifteen sensing devices (Raspberry Pis) established session keys with each other. Fig. \ref{sessionkey} shows the final computed session key between the IoT device $D_x$ and the IoT device $D_y$. 
\begin{figure}[H]
\includegraphics[width=\linewidth, height=4cm]{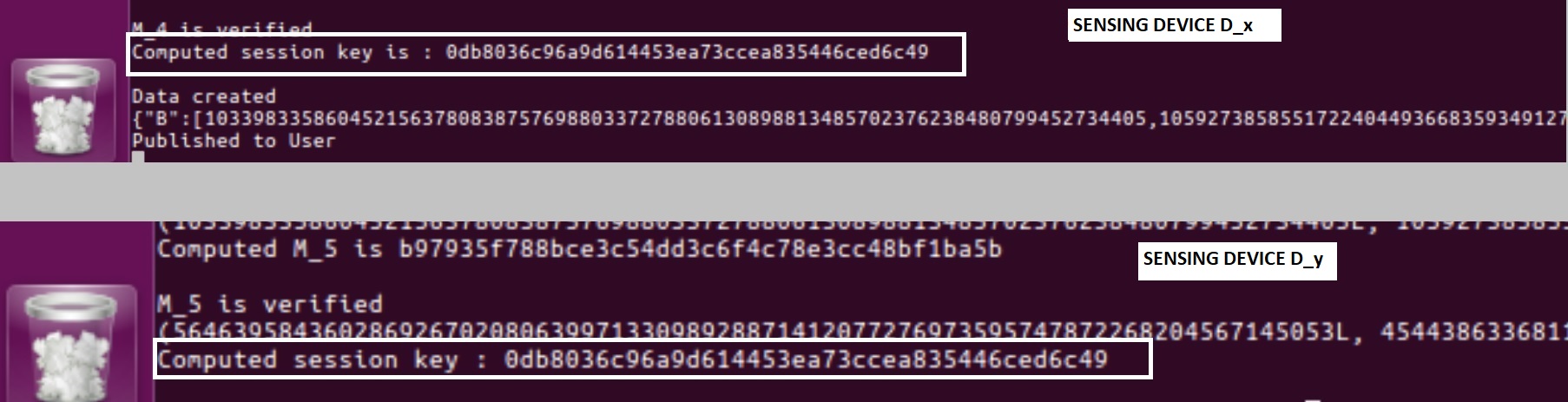}
\caption{Session key computation} \label{sessionkey}
\end{figure}

\noindent  The MQTT protocols works with three quality of services, \textit{QoS 0} (at most once), \textit{QoS 1} (at least once) and \textit{QoS 2} (exactly once) for packet transmissions. As mentioned earlier, we collected average throughput, packet delivery ratio, and round-trip delay for the setup by analysing the data collected using a wireshark tool. We define the average throughput as an average number of packets transmitted and successfully received in a unit. We observed that the average throughput for the proposed setup was 643 packets per minute. The average packet delivery ratio was around 99.34\%, with 0.66 packet loss. The average packet delivery ratio might reduce if we use a global broker. The range of round-trip delay (from $D_x$ to TA, TA to $D_y$, $D_y$ to TA, and $TA$ to $D_x$) was around 45 ms - 70 ms because of the light-weight computation of the proposed EBAKE-SE protocol.       
\section{Comparative analysis}
\label{comparative}
\noindent In this section, we highlight the computational efficiency of the proposed scheme using a comparative analysis of it based on the number of cryptographic operations, computation time (in ms), and communication cost (in bits). We compare the proposed EBAKE-SE with the other recently proposed schemes for a similar environment. 
\subsection{Cryptographic Operations}
\noindent The Table \ref{Table:2} highlights a comparative analysis of the proposed scheme (only authentication phase) with other existing schemes based on the number of cryptographic operations required. 
\begin{table}[H]
    \centering
    \begin{threeparttable}
    \caption{Network Model and Cryptographic Operations}
    \begin{tabular}{|p{0.8cm}|p{1.2cm}|p{0.4cm}|p{0.4cm}|p{0.4cm}|p{0.4cm}|p{0.4cm}|p{0.4cm}|} \hline
    \small{\textbf{Scheme}} & \small{\textbf{Model}} & \small{\textbf{$OP_1$}} & \small{\textbf{$Op_2$}} & \small{\textbf{$Op_3$}} & \small{\textbf{$Op_4$}} & \small{\textbf{$Op_5$}} & \small{\textbf{$Op_6$}} \\ \hline
    Ours & D-TA-D & 2 & 4 & 11 & 2 & - & - \\
    \cite{Das2019} & D-D & - & - & 12 & - & 12 & - \\ 
    \cite{li2017robust} & U-G-D & - & - & 19 & 9 & 6 & - \\
    \cite{wu2018secure} & MD-MD-S & 3 & - & 9 & 2 & 13 & -  \\
    \cite{li2019} & U-G-D & - & 1 & 22 & 11 & 6 & - \\ \hline
    \end{tabular}
    \begin{tablenotes}
      \small
      \item \textcolor{black}{\textit{$OP^1$}: Symmetric Encryption/Decryption, \textit{$OP^2$}: Asymmetric Encryption/Decryption, \textit{$OP^3$}: Hash function,} \textit{$OP^4$}: XOR operation, \textit{$OP^5$}: ECC point multiplication operation, \textit{$OP^6$}: ECC point summation operations, \textit{U}: User, \textit{GW}:Gateway,  \textit{TA}:Trusted Authority, \textit{D}: Sensing device, \textit{S}: Server, 
      \textit{MD}: Mobile device.
    \end{tablenotes}
    \label{Table:2}
    \end{threeparttable}
\end{table}
\subsection{Computation Time}
\noindent Table \ref{Table:3} highlights a comparative analysis of EBAAC-SE with other existing schemes based on the computation time required by the scheme. In the initial phase of our implementation, we collected results for basic cryptographic operations. These results are collected for the environment discussed in section \ref{implelment}. Observations of these computations were as follows: The time required for the single hash function using SHA was ($T_h$) 0.043 ms. The time required by a single elliptic curve point addition operation was ($T_{pa}$) 0.068 ms. The time required by a single elliptic curve point multiplication operation was ($T_{pm}$) 12.226 ms. The time required for single symmetric encryption over AES was ($T_{sym}$) 0.046ms. The time required by single ECC encryption is ($T_{asym}$ $\approx$ $T_{pm}$) 12.268 ms. Based on these observations, in Table \ref{Table:3}, we highlight a computation time-based comparison between the proposed scheme and other existing schemes. 
\begin{table*}
    \centering
    \caption{Computation time}
    \begin{tabular}{p{1cm}p{3cm}p{2cm}p{3cm}p{4cm}p{2cm}} \hline
    \textbf{Scheme} & \textbf{Device-1} & \textbf{Gateway/ cloud} & \textbf{Device-2} & \textbf{Total operations} & \textbf{Total time}\\ \hline
    Proposed & $T_{sym}$ + $2T_{asym}$ + $3T_h$ & $T_{sym}$ + $5T_h$ & $2T_{asym}$ + $3T_h$ & 2$T_{sym}$ + $4T_{asym}$ + $11T_h$ & 49.469 ms \\
    \cite{Das2019} & 6$T_{pm}$ + 6$T_{h}$ + 2$T_{pa}$ & - & 6$T_{pm}$ + 6$T_{h}$ + 2$T_{pa}$ & 12$T_{pm}$ + 12$T_{h}$ + 4$T_{pa}$ &  147.5 ms\\ 
    \cite{challa2017secure} & 5$T_{pm}$ + 4$T_h$ & 4$T_{pm}$ + 3$T_h$ & 14$T_{pm}$ + 12$T_h$ & 23$T_{pm}$ + 19$T_h$ & 171.68 ms \\
    \cite{li2017robust} & 3$T_{pm}$ + 8$T_h$ & $T_{pm}$ + 7$T_h$ & 2$T_{pm}$ + 4$T_h$ & 6$T_{pm}$ + 19$T_h$ & 74.173 ms\\
    \cite{wu2018secure} & 6$T_{pm}$ + 4$T_h$ + 2$T_{sym}$ & 4$T_{pm}$ + 5$T_h$ & 3$T_{pm}$ + $T_{sym}$ & 13$T_{pm}$ + 9$T_h$ + 3$T_{sym}$ & 159.454 ms \\ 
    \cite{li2019} & 3$T_{pm}$ + 10$T_h$ + $T_{asym}$ & $T_{pm}$ + 8$T_h$ & 2$T_{pm}$ + 8$T_h$ & 6$T_{pm}$ + 22$T_h$ + $T_{asym}$ & 86.528 ms \\ \hline   
    \end{tabular}
    \label{Table:3}
\end{table*}
\section{Conclusions and Future work}
\label{conclusion}
\noindent In this paper, we proposed an ECC-based authenticated key exchange scheme between two Industrial IoT devices via trusted authority. We use a tamper-proof microprocessor called a secret element (SE) to store the secret parameters of sensing devices. We provided cryptanalysis for the RUA scheme proposed by Das et al. for a similar environment and highlighted numerous vulnerabilities such as MITM attack and impersonation attack. Afterwards, we offered an RUA using ECC between two advanced-IoT devices via cloud trusted authority. We presented informal security analysis as well as formal analysis for EBAKE-SE. We compared the presented EBAKE-SE with existing schemes based on security features, computation time, and several cryptography operations. Furthermore, we presented an implementation environment using the publish-subscribe-based MQTT protocol. \textcolor{black}{The numerous IoT-based industries (such as smart homes, smart healthcare, smart transport, smart security, and surveillance system) can use the proposed EBAKE-SE to enhance their security mechanism with acceptable reliability and efficiency.}
\section*{Acknowledgements}
\noindent This work was supported by the Researchers Supporting Project (No. RSP-2021/395), King Saud University, Riyadh, Saudi Arabia.).
\bibliographystyle{elsarticle-num}
\bibliography{ref}

\begin{thebibliography}{10}
\expandafter\ifx\csname url\endcsname\relax
  \def\url#1{\texttt{#1}}\fi
\expandafter\ifx\csname urlprefix\endcsname\relax\def\urlprefix{URL }\fi
\expandafter\ifx\csname href\endcsname\relax
  \def\href#1#2{#2} \def\path#1{#1}\fi

\bibitem{wang2021sparse}
H.~Wang, X.~Li, R.~H. Jhaveri, T.~R. Gadekallu, M.~Zhu, T.~A. Ahanger, S.~A.
  Khowaja, Sparse bayesian learning based channel estimation in fbmc/oqam
  industrial iot networks, Computer Communications.

\bibitem{extra4}
M.~Alkhelaiwi, W.~Boulila, J.~Ahmad, A.~Koubaa, M.~Driss, An efficient approach
  based on privacy-preserving deep learning for satellite image classification,
  Remote Sensing 13~(11).
\newblock \href {http://dx.doi.org/10.3390/rs13112221}
  {\path{doi:10.3390/rs13112221}}.

\bibitem{jhaveri2021fault}
R.~Jhaveri, R.~Sagar, G.~Srivastava, T.~R. Gadekallu, V.~Aggarwal,
  Fault-resilience for bandwidth management in industrial software-defined
  networks, IEEE Transactions on Network Science and Engineering.

\bibitem{extra3}
M.~Driss, D.~Hasan, W.~Boulila, J.~Ahmad, Microservices in iot security:
  Current solutions, research challenges, and future directions, Procedia
  Computer Science 192 (2021) 2385--2395, knowledge-Based and Intelligent
  Information \& Engineering Systems: Proceedings of the 25th International
  Conference KES2021.
\newblock \href {http://dx.doi.org/https://doi.org/10.1016/j.procs.2021.09.007}
  {\path{doi:https://doi.org/10.1016/j.procs.2021.09.007}}.

\bibitem{ali2019fractal}
A.~Ali, H.~Rafique, T.~Arshad, M.~A. Alqarni, S.~H. Chauhdary, A.~K. Bashir, A
  fractal-based authentication technique using sierpinski triangles in smart
  devices, Sensors 19~(3) (2019) 678.

\bibitem{Sobin2020}
C.~C. Sobin, \href{https://doi.org/10.1007/s11277-020-07108-5}{A survey on
  architecture, protocols and challenges in iot}, Wireless Personal
  Communications 112~(3) (2020) 1383--1429.
\newblock \href {http://dx.doi.org/10.1007/s11277-020-07108-5}
  {\path{doi:10.1007/s11277-020-07108-5}}.
\newline\urlprefix\url{https://doi.org/10.1007/s11277-020-07108-5}

\bibitem{Atzori2010}
L.~Atzori, A.~Iera, G.~Morabito,
  \href{http://dx.doi.org/10.1016/j.comnet.2010.05.010}{{The Internet of
  Things: A survey}}, Computer Networks, Elsevier 54~(15) (2010) 2787--2805.
\newblock \href {http://arxiv.org/abs/arXiv:1011.1669v3}
  {\path{arXiv:arXiv:1011.1669v3}}, \href
  {http://dx.doi.org/10.1016/j.comnet.2010.05.010}
  {\path{doi:10.1016/j.comnet.2010.05.010}}.
\newline\urlprefix\url{http://dx.doi.org/10.1016/j.comnet.2010.05.010}

\bibitem{Fuqaha2015}
A.~Al-fuqaha, S.~Member, M.~Guizani, M.~Mohammadi, S.~Member, {Internet of
  Things : A Survey on Enabling}, IEEE Communication 17~(4) (2015) 2347--2376.

\bibitem{Sethi2017}
P.~Sethi, S.~R. Sarangi, {Internet of Things: Architectures, Protocols, and
  Applications}, Journal of Electrical and Computer Engineering, Hindawi 2017.
\newblock \href {http://dx.doi.org/10.1155/2017/9324035}
  {\path{doi:10.1155/2017/9324035}}.

\bibitem{extra1}
Y.~Hajjaji, W.~Boulila, I.~R. Farah, I.~Romdhani, A.~Hussain, Big data and
  iot-based applications in smart environments: A systematic review, Computer
  Science Review 39 (2021) 100318.
\newblock \href
  {http://dx.doi.org/https://doi.org/10.1016/j.cosrev.2020.100318}
  {\path{doi:https://doi.org/10.1016/j.cosrev.2020.100318}}.

\bibitem{extra2}
S.~B. Atitallah, M.~Driss, W.~Boulila, H.~B. Ghézala, Leveraging deep learning
  and iot big data analytics to support the smart cities development: Review
  and future directions, Computer Science Review 38 (2020) 100303.
\newblock \href
  {http://dx.doi.org/https://doi.org/10.1016/j.cosrev.2020.100303}
  {\path{doi:https://doi.org/10.1016/j.cosrev.2020.100303}}.

\bibitem{Maple2017}
C.~Maple, {Security and privacy in the internet of things}, Journal of Cyber
  Policy 2~(2) (2017) 155--184.
\newblock \href {http://dx.doi.org/10.1080/23738871.2017.1366536}
  {\path{doi:10.1080/23738871.2017.1366536}}.

\bibitem{stoyanova2020survey}
M.~Stoyanova, Y.~Nikoloudakis, S.~Panagiotakis, E.~Pallis, E.~K. Markakis, A
  survey on the internet of things (iot) forensics: Challenges, approaches and
  open issues, IEEE Communications Surveys \& Tutorials.

\bibitem{irshad2020fuzzy}
A.~Irshad, M.~Usman, S.~A. Chaudhry, A.~K. Bashir, A.~Jolfaei, G.~Srivastava,
  Fuzzy-in-the-loop-driven low-cost and secure biometric user access to server,
  IEEE Transactions on Reliability.

\bibitem{MOHAMADNOOR2019}
M.~binti [Mohamad~Noor], W.~H. Hassan, Current research on internet of things
  (iot) security: A survey, Computer Networks 148 (2019) 283 -- 294.
\newblock \href
  {http://dx.doi.org/https://doi.org/10.1016/j.comnet.2018.11.025}
  {\path{doi:https://doi.org/10.1016/j.comnet.2018.11.025}}.

\bibitem{arul2018console}
R.~Arul, G.~Raja, A.~K. Bashir, J.~Chaudry, A.~Ali, A console grid leveraged
  authentication and key agreement mechanism for lte/sae, IEEE Transactions on
  Industrial Informatics 14~(6) (2018) 2677--2689.

\bibitem{Das2019}
A.~K. {Das}, M.~{Wazid}, A.~R. {Yannam}, J.~J. P.~C. {Rodrigues}, Y.~{Park},
  Provably secure ecc-based device access control and key agreement protocol
  for iot environment, IEEE Access 7 (2019) 55382--55397.

\bibitem{qureshi2021stream}
N.~M.~F. Qureshi, I.~F. Siddiqui, A.~Abbas, A.~K. Bashir, C.~S. Nam, B.~S.
  Chowdhry, M.~A. Uqaili, Stream-based authentication strategy using iot sensor
  data in multi-homing sub-aqueous big data network, Wireless Personal
  Communications 116~(2) (2021) 1217--1229.

\bibitem{Sethia2018}
D.~{Sethia}, D.~{Gupta}, H.~{Saran}, Nfc secure element-based mutual
  authentication and attestation for iot access, IEEE Transactions on Consumer
  Electronics 64~(4) (2018) 470--479.
\newblock \href {http://dx.doi.org/10.1109/TCE.2018.2873181}
  {\path{doi:10.1109/TCE.2018.2873181}}.

\bibitem{patel2020enhanced}
C.~Patel, D.~Joshi, N.~Doshi, A.~Veeramuthu, R.~Jhaveri, An enhanced approach
  for three factor remote user authentication in multi-server environment,
  Journal of Intelligent \& Fuzzy Systems~(Preprint)  1--12.

\bibitem{miller1985use}
V.~S. Miller, Use of elliptic curves in cryptography, in: Conference on the
  theory and application of cryptographic techniques, Springer, 1985, pp.
  417--426.

\bibitem{koblitz1987elliptic}
N.~Koblitz, Elliptic curve cryptosystems, Mathematics of computation 48~(177)
  (1987) 203--209.

\bibitem{dhillon2019secure}
P.~K. Dhillon, S.~Kalra, Secure and efficient ecc based sip authentication
  scheme for voip communications in internet of things, Multimedia Tools and
  Applications 78~(16) (2019) 22199--22222.

\bibitem{kumar2019secure}
D.~Kumar, H.~S. Grover, et~al., A secure authentication protocol for wearable
  devices environment using ecc, Journal of Information Security and
  Applications 47 (2019) 8--15.

\bibitem{lohachab2019ecc}
A.~Lohachab, et~al., Ecc based inter-device authentication and authorization
  scheme using mqtt for iot networks, Journal of Information Security and
  Applications 46 (2019) 1--12.

\bibitem{qi2019secure}
M.~Qi, J.~Chen, Y.~Chen, A secure authentication with key agreement scheme
  using ecc for satellite communication systems, International Journal of
  Satellite Communications and Networking 37~(3) (2019) 234--244.

\bibitem{garg2019towards}
S.~Garg, K.~Kaur, G.~Kaddoum, K.-K.~R. Choo, Towards secure and provable
  authentication for internet of things: Realizing industry 4.0, IEEE Internet
  of Things Journal.

\bibitem{dammak2019token}
M.~Dammak, O.~R.~M. Boudia, M.~A. Messous, S.~M. Senouci, C.~Gransart,
  Token-based lightweight authentication to secure iot networks, in: 2019 16th
  IEEE Annual Consumer Communications \& Networking Conference (CCNC), IEEE,
  2019, pp. 1--4.

\bibitem{dang2020pragmatic}
T.~K. Dang, C.~D. Pham, T.~L. Nguyen, A pragmatic elliptic curve
  cryptography-based extension for energy-efficient device-to-device
  communications in smart cities, Sustainable Cities and Society 56 (2020)
  102097.

\bibitem{li2017robust}
X.~Li, J.~Niu, M.~Z.~A. Bhuiyan, F.~Wu, M.~Karuppiah, S.~Kumari, A robust
  ecc-based provable secure authentication protocol with privacy preserving for
  industrial internet of things, IEEE Transactions on Industrial Informatics
  14~(8) (2017) 3599--3609.

\bibitem{kocher1999}
P.~Kocher, J.~Jaffe, B.~Jun, Differential power analysis, in: Annual
  international cryptology conference, Springer, 1999, pp. 388--397.

\bibitem{secureelement}
What is an iot hardware secure element?,
  \url{https://cerberus-laboratories.com/blog/$iot_hsms$/}, accessed:
  2020-12-20.

\bibitem{challa2017secure}
S.~Challa, M.~Wazid, A.~K. Das, N.~Kumar, A.~G. Reddy, E.-J. Yoon, K.-Y. Yoo,
  Secure signature-based authenticated key establishment scheme for future iot
  applications, IEEE Access 5 (2017) 3028--3043.

\bibitem{wu2018secure}
L.~Wu, J.~Wang, K.-K.~R. Choo, D.~He, Secure key agreement and key protection
  for mobile device user authentication, IEEE Transactions on Information
  Forensics and Security 14~(2) (2018) 319--330.

\bibitem{porambage2015group}
P.~Porambage, A.~Braeken, C.~Schmitt, A.~Gurtov, M.~Ylianttila, B.~Stiller,
  Group key establishment for enabling secure multicast communication in
  wireless sensor networks deployed for iot applications, IEEE Access 3 (2015)
  1503--1511.

\bibitem{li2019}
X.~Li, J.~Peng, M.~S. Obaidat, F.~Wu, M.~K. Khan, C.~Chen, A secure
  three-factor user authentication protocol with forward secrecy for wireless
  medical sensor network systems, IEEE Systems Journal 14~(1) (2019) 39--50.

\end{thebibliography}
\end{document}